\documentclass[10pt, a4paper, twocolumn]{IEEEtran}

\usepackage{amsmath,amsthm,epsfig,amssymb,verbatim,amsopn,cite,subfigure,multirow}
\usepackage{balance,nopageno}
\usepackage[usenames,dvipsnames]{color}
\usepackage[all]{xy}  
\usepackage{url}
\usepackage{amsfonts}
\usepackage{amssymb}
\usepackage{epsfig}
\usepackage{epstopdf}
\usepackage{slashbox}
\usepackage{bm}
\usepackage{balance}
\usepackage[margin=16.0mm,top=17mm,bottom=16mm]{geometry}

\newtheorem{Lemma}{Lemma}

\newtheorem{Corollary}[Lemma]{Corollary}

  {\proof}{\proofend}
\newtheorem{proposition}{Proposition}

\DeclareMathOperator*{\argmax}{arg\,max}

\newcommand{\ve}[1]{\boldsymbol{#1}}
\newcommand{\E}[1]{\mathbb{E}\left\{#1\right\}}

\newcommand{\vH}{\ve{H}} \newcommand{\vh}{\ve{h}}

\newcommand{\qg}{{\bf g}}
\newcommand{\qh}{{\bf h}}

\newcommand{\qn}{{\bf n}}

\newcommand{\qw}{{\bf w}}

\newcommand{\qA}{{\bf A}}

\newcommand{\qH}{{\bf H}}
\newcommand{\qI}{{\bf I}}

\newcommand{\SINRuA}{\mathsf{SINR_u^A}}
\newcommand{\SINRdA}{\mathsf{SINR_d^A}}
\newcommand{\SINRuS}{\mathsf{SINR_u^S}}
\newcommand{\SINRdS}{\mathsf{SINR_d^S}}
\newcommand{\SNRd}{\mathsf{SNR_d}}
\newcommand{\SNRu}{\mathsf{SNR_u}}

\newcommand{\Us}{\mathsf{u}}
\newcommand{\Ds}{\mathsf{d}}
\newcommand{\AP}{\mathsf{a}}

\newcommand{\MRC}{\mathsf{MRC}}
\newcommand{\MRT}{\mathsf{MRT}}
\newcommand{\ZF}{\mathsf{ZF}}
\newcommand{\FD}{\mathsf{FD}}
\newcommand{\RFD}{\mathcal{R}^\FD_{\mathsf{sum}}}
\newcommand{\RHDs}{R^\mathsf{{HD}}_{\mathsf{{sum}}}}

\newcommand{\Phd}{\Phi_{\Ds}}

\newcommand{\Galf}{\mathcal{G}(M,\delta,p\lambda)}
\newcommand{\Galfu}{\mathcal{G}(M,\delta,(1-p)\lambda)}

\newcommand{\Sap}{\sigma_{\AP\AP}^2}

\newcommand{\Prob}{\textnormal{Pr}}

\newcommand{\AuthorOne}{Mohammadali Mohammadi$^\dag$}
\newcommand{\AuthorTwo}{Himal A. Suraweera$^\S$}

\newcommand{\AuthorFour}{Chintha Tellambura$^*$}

\newcommand{\be}{\begin{equation}} \newcommand{\ee}{\end{equation}}
\newcommand{\bea}{\begin{eqnarray}} \newcommand{\eea}{\end{eqnarray}}

\usepackage{multibib}
\usepackage[nodisplayskipstretch]{setspace}
\newcites{Prim}{Very important papers}

\definecolor{light-gray}{gray}{0.65}

\newcounter{mytempeqcounter}

\title{Full-Duplex Cloud-RAN with Uplink/Downlink Remote Radio Head Association}

\author{\authorblockN{\AuthorOne, \AuthorTwo, and \AuthorFour}\\
\small{
$^\dag$Faculty of  Engineering, Shahrekord University, Iran (e-mail: m.a.mohammadi@eng.sku.ac.ir)\\
$^\S$Department of Electrical and Electronic Engineering, University of Peradeniya, Sri Lanka (e-mail: himal@ee.pdn.ac.lk)\\
$^*$Department of Electrical and Computer Engineering, University of Alberta, Canada (e-mail: chintha@ece.ualberta.ca)
}}\normalsize

\begin{document}

\maketitle
\thispagestyle{empty}

\begin{abstract}
This paper considers a cloud radio access network (C-RAN) where spatially distributed remote radio heads (RRHs) communicate with a full-duplex user. In order to reflect a realistic scenario, the uplink (UL) and downlink (DL) RRHs are assumed to be equipped with multiple antennas and distributed according to a Poisson point process. We consider all participate and nearest RRH association schemes with distributed beamforming in the form of maximum ratio combining/maximal ratio transmission (MRC/MRT) and zero-forcing/MRT(ZF/MRT) processing. We derive analytical expressions useful to compare the average sum rate among association schemes as a function of the number of RRHs antennas and density of the UL and DL RRHs. Numerical results show that significant performance improvements can be achieved by using the full-duplex mode as compared to the half-duplex mode, while the choice of the beamforming design as well as the RRH association scheme plays a critical role in determining the full-duplex gains. 
\end{abstract}
\section{Introduction}
Cloud radio access network (C-RAN) is a new network paradigm acclaimed to become a key integral component of future 5G radio access technology~\cite{Chintha:IWC:2015,Peng:WCL:2014,VPoor:SPL:2013}. C-RAN architecture can provide high energy-efficiency transmission, improved spectral utilization and reduce capital/operating expenses for cellular network deployment. For these reasons, C-RAN concept has become a topic of interest to researchers and mobile operators~\cite{Ratnarajah:TSP:2015}. The main idea of C-RAN is to deploy a pool of distributed radio units called remote radio heads (RRHs) for signal transmission/collection operations and connect them with a centrally located baseband unit capable of sophisticated processing via a high speed optical backbone.

On parallel, full-duplex communication capable of boosting the spectral efficiencies of current 4G wireless systems shows high promise as a complementary approach with C-RAN for 5G implementation~\cite{Zhang:WMC:2014,Sabharwal:JSAC:2014}. Full-duplex radio nodes can transmit and receive on the same channel. There has been rapid progress made in last few years on both theory and experimental hardware design to make full-duplex operation an efficient practical solution~\cite{Duarte:PhD:dis,Katti:Sigcomm:2013}. To this end, a major challenge to overcome in full-duplex implementation is the signal leakage from the output of the transceiver to the input. This form of interference, called the loopback interference (LI), if not mitigated substantially, can cause significant performance degradation~\cite{Riihonen:JSP:2011}. Traditionally, LI suppression is performed in the antenna domain using passive techniques such as the use of electromagnetic shields, directional antennas and antenna separation. When full-duplex and C-RAN are combined, path loss naturally serves a simple effective phenomenon for LI suppression since RRHs will be distributed.

There have been several studies that have harnessed tools from stochastic geometry to analyze the performance of C-RANs with randomly located RRHs. In~\cite{Andrews:WCOM:2008}, a binomial point process and a Poisson point process (PPP) was considered to model antenna and user distributions of a C-RAN. The authors developed an analytical framework to analyze best antenna and channel selection with fading and shadowing effects. The ergodic capacity of a multi-cell distributed RRH system has been studied in~\cite{Yicheng:WCOM:2014}. In~\cite{Peng:WCL:2014}, the outage probability and the ergodic capacity achieved with RRH association strategies for C-RANs were characterized. In order to investigate the performance of distributed antenna arrays, beamforming and base station selection was compared in~\cite{VPoor:SPL:2013}. In~\cite{Liu:TSP:2014}, average weighted sum-rate maximization under antenna selection and transmit power constraints has been carried out assuming regularized zero forcing (ZF). In~\cite{Ratnarajah:TSP:2015}, the downlink (DL) transmission of a multiple antenna equipped C-RAN network with maximal ratio transmission (MRT) or transmit antenna selection has been analyzed. In these previous works~\cite{Yicheng:WCOM:2014,Ratnarajah:TSP:2015,Andrews:WCOM:2008}, only UL or DL performance have been considered. Full-duplex operation with distributed antennas was proposed in~\cite{Leung:TVT}. However, it neglected an important aspect of full-duplex operation, namely, perfect LI cancellation was assumed.

This paper considers a C-RAN with full-duplex transmission. We consider a case in which multiple antenna equipped RRHs communicate with a full-duplex user to support simultaneous UL and DL transmissions. Our contributions are summarized as follows:
\begin{itemize}
\item Assuming different UL and DL linear decoding and precoding schemes, namely, maximum ratio combining (MRC)/MRT and ZF/MRT, we derive exact and tractable expressions for the average UL and DL rate of the full-duplex user for the single UL/DL RRH association (SRA) scheme.
\item We show that all RRH association (ARA) scheme results in a rate region that is strongly biased toward UL or DL, but using SRA scheme results in a more balanced rate region.
\item Our findings reveal that for a fixed value of LI power, the ZF/MRT scheme can ensure a balance between maximizing the system average sum rate and maintaining acceptable level of fairness between the UL/DL transmission. Moreover, we compare the performance of full-duplex and half-duplex modes under ARA and SRA schemes to show the benefits of the former.
\end{itemize}
\emph{Notation:} We use bold upper case letters to denote matrices, bold lower case letters to denote vectors. $\|\cdot\|$  and $(\cdot)^{\dag}$ denote the Euclidean norm and conjugate transpose operator, respectively; ${\tt E}\left\{x\right\}$ stands for the expectation of the random variable (RV) $x$;
$f_X(\cdot)$ and $F_X(\cdot)$ denote the probability density function (pdf) and cumulative distribution function (cdf) of the RV $X$, respectively; $\mathcal{M}_X(s)$ is the moment generating function (MGF) of the RV, $X$; 
$\Gamma(a)$ is the Gamma function; $\Gamma(a,x)$ is upper incomplete Gamma function~\cite[Eq. (8.310.2)]{Integral:Series:Ryzhik:1992};
and $G_{p q}^{m n} \left( z \  \vert \  {a_1\cdots a_p \atop b_1\cdots b_q} \right)$ denotes the Meijer G-function~\cite[ Eq. (9.301)]{Integral:Series:Ryzhik:1992}.

\section{System Model}

We consider a C-RAN, consisting of baseband unit (BBU) and a group of spatially distributed RRHs  to jointly support a full-duplex user, denoted by $U$ for both DL and UL transmissions. We assume that each RRH, is equipped with $M \geq 1$  antennas, and the full-duplex user is equipped with two antennas: one receive antenna and one transmit antenna. The locations of the RRHs are modeled as a homogeneous PPP $\Phi=\{x_k\}$ with density $\lambda$ in a disc $\mathcal{D}$, of radius $R$. We assume that $p\%$ of the RRHs, are deployed to assist the DL communication and $(1-p)\%$ for UL communication. Therefore, the set of DL RRHs is denoted as $\Phi_{\Ds} =\{x_k\in \Phi: B_k(p)=1\}$ where $B_k(p)$  are independent and identically distributed (i.i.d.) Bernoulli RVs with parameter $p$ associated with $x_k$. Similarly, the set of UL RRHs is a PPP with density $(1-p)\lambda$ and is denoted as $\Phi_{\Us} =\{x_k\in \Phi: B_k(p)=0\}$. Therefore, the number of DL RRHs, $N_\Ds$ and the number of UL RRHs, $N_\Us$  in $\mathcal{D}$ are Poisson distributed as $\Prob(N_i) = \mu_i^{N_i}e^{-\mu_i}/\Gamma(N_i+1)$, with $i\in\{\Us, \Ds\}$, $\mu_\Ds = \pi p\lambda R^2$ and $\mu_\Us = \pi (1-p)\lambda R^2$.

\subsection{Channel Model}
Signal propagation is subject to both small-scale multipath fading and large-scale path loss. We denote the DL channel vector from RRH $i$ to $U$ as $\qh_i\in \mathbb{C}^{M \times 1}$ and the UL channel vector from $U$ to RRH $i$ as $\qg _i^{\dag}\in \mathbb{C}^{1 \times M}$, respectively.  These channels capture the small-scale fading and are modeled as Rayleigh fading such that $\qg _i$  and $\qh_i \sim  \mathcal{CN}(\textbf{0}_M, \qI_M )$,  where $\mathcal{CN}(\cdot,\cdot )$,  denotes a circularly symmetric complex Gaussian distribution. The path loss model is denoted by $\ell(\cdot): \mathbb{R}^2 \rightarrow \mathbb{R}^+$.  We consider a non-singular path loss model with $\ell(x_1,x_2)=\frac{1}{\epsilon + \|x_1-x_2\|^\alpha}$ where $\alpha>2$ is the path loss exponent and $\epsilon>0$ is the reference distance. Further, as in \cite{Ratnarajah:TSP:2015} we assume that there exist an ideal low-latency backhaul network with sufficiently large capacity (e.g. optical fiber) connecting the set of RRHs to the BBU, which performs all the baseband signal processing and transmission scheduling for all RRHs.

\subsection{Association Schemes}
For the system under consideration, we investigate the performance of the following two RRH association schemes:
\begin{itemize}
\item \emph{All RRH Association (ARA) Scheme}: All corresponding DL RRHs cooperatively transmit the signal, $s_{\Ds}$ to the full-duplex User, $U$. Moreover, all the corresponding UL RRHs deliver signals from $U$ to the BBU.
\item \emph{Single Nearest RRH Association (SRA) Scheme}: The full-duplex User, $U$ associates with the nearest DL RRH and the nearest UL RRH, respectively. Without loss of generality, we assume that the full-duplex user, $U$ is located at the origin of $\mathcal{D}$.  Therefore, the associated UL RRH $p$ and DL RRH $q$ for user  $U$ are given by $p=\argmax_{i\in \Phi_{\mathsf{u}}}\ell(x_{i})$ and $q=\argmax_{i\in \Phi_{\mathsf{d}}}\ell(x_{i})$, respectively.\footnote{Our results can also be easily extended to an $N$ nearest RRH association scheme, where User $U$ associates with the $N$ nearest DL and UL RRHs among the total $N_\Ds$ ($N_\Us$) DL (UL) RRHs.}
\end{itemize}

We point out that in the case of full-duplex transmission, selection of a nearest RRH is also a practical assumption, since transmitting high power signals towards (from) distant periphery UL (DL) RRHs in order to guarantee a quality-of-service can cause overwhelming LI at $U$ (interference between UL and DL RRHs). Similar C-RAN association schemes in context of a half-duplex user can also be found in~\cite{Peng:WCL:2014,Ratnarajah:TSP:2015}.
\subsection{Uplink/Downlink Transmission}
\emph{DL Transmission:}
We assume that all DL RRHs transmit with power $P_b$ as in \cite{Ratnarajah:TSP:2015}. Hence, according to the ARA scheme, the received signal at the user can be expressed as
\vspace{-0.2em}
\be\label{eq:rx at user}
y_\Ds  = \sum_{i \in \Phi_\Ds \cap b(o,R)} \sqrt{P_b \ell(x_{i})}\qh_{i}^{\dag}\qw_{t,i} s_{\Ds} \!+\! \sqrt{P_u}h_{\mathsf{LI}}s_u + n_\Ds,
\ee
where $b(o,R)$ denotes a ball of radius $R$ centered at the origin, $\qw_{t,i}\in\mathbb{C}^{M\times 1}$ is the transmit beamforming vector at DL RRH $i$, $P_u$ is the user transmit power and $s_u$  is the user signal satisfying ${\tt E}\left\{s_u s_u^{\dag}\right\}=1$, and $n_\Ds$ denotes the additive white Gaussian noise (AWGN) with zero mean. We proceed with all noise variances set to one.
$h_{\mathsf{LI}}$ denotes the LI channel at the user.  In order to mitigate the adverse effects of the LI on system's performance, an interference cancellation scheme (i.e. analog/digital cancellation) can be used at the full-duplex user and we model the residual LI channel with Rayleigh fading assumption since the strong line-of-sight component can be estimated and removed~\cite{Duarte:PhD:dis,Riihonen:JSP:2011}. Since each implementation of a particular analog/digital LI cancellation scheme can be characterized by a specific residual power, a parameterization by $h_{\mathsf{LI}}$ satisfying $\E{|h_{\mathsf{LI}}|^2}=\Sap$  allows these effects to be studied in a generic way~\cite{Riihonen:JSP:2011}.

By invoking~\eqref{eq:rx at user}, the DL signal-to-interference-plus-noise ratio (SINR) for the user is given by
\vspace{-0.2em}
\begin{align}\label{eq:SINRd AR}
 \SINRdA &= \frac{\sum_{i \in \Phi_\Ds \cap b(o,R)} P_b \ell(x_{i})|\qh_{i}^{\dag}\qw_{t,i}|^2  }{{P_u}|h_{\mathsf{LI}}|^2 + 1}.
\end{align}

Moreover, with the SRA scheme, the received SINR at the user can be established as
\vspace{-0.2em}
\begin{align}\label{eq:SINRd SR}
 \SINRdS&= \frac{P_b \ell(x_{q})|\qh_{q}^{\dag}\qw_{t,i}|^2  }{{P_u}|h_{\mathsf{LI}}|^2+ 1}.
\end{align}
\emph{UL Transmission:} Let us denote $\qw_{r,j}\in\mathbb{C}^{M\times 1}$ as the receive beamforming vector at the UL RRH, $j$. According to the ARA scheme, received signal at the BBU is given by
\vspace{-0.2em}
\begin{align}\label{eq:rx at user}
 y_u&=\sum_{j\in \Phi_{\mathsf{u}} \cap b(o,R)} \Big(\sqrt{P_u\ell(x_{j})}\qw_{r,j}^{\dag}\qg_{j} x_u
 \\&\quad
 +\sum_{i\in \Phi_{\mathsf{d}} \cap b(o,R)} \sqrt{P_b \ell(x_j,x_{i})}\qw_{r,j}^{\dag}\qH_{\mathsf{ud}}^{ji}\qw_{t,i}s_{\Ds}+ \qw_{r,j}^{\dag}\qn_j\Big),\nonumber
\end{align}
where $\qH_{\mathsf{ud}}^{ji}\in \mathbb{C}^{M \times M}$  is the channel matrix between the DL RRH $i$ and UL RRH $j$  consists of complex Gaussian distributed entries with zero mean and unit variance, $\qn_j\sim  \mathcal{CN}(\textbf{0}_M, \qI_M )$  denotes the AWGN vector at the  UL RRH $j$. Therefore, the SINR can be expressed as
\vspace{-0.2em}
\begin{align}\label{eq:SINRu AR}
 \SINRuA&=
 \frac{\sum_{j\in \Phi_{\mathsf{u}}\cap b(o,R)} P_u \ell(x_{j})|\qw_{r,j}^{\dag}\qg_{j}|^2}
                       { I_{\mathsf{ud}}+ \|\qw_{r,j}\|^2},
\end{align}
where
\vspace{-0.2em}
\begin{align*}
I_{\mathsf{ud}}=\sum_{j\in \Phi_{\mathsf{u}}\cap b(o,R)}\:\sum_{i\in \Phi_{\mathsf{d}}\cap b(o,R)}\! P_b \ell(x_j,x_{i})|\qw_{r,j}^{\dag}\qH_{\mathsf{ud}}^{ji}\qw_{t,i}|^2.
\end{align*}

According to the SRA scheme only one UL (nearest) RRH and one DL (nearest) RRH are selected to assist the full-duplex user. Let the sub-indexes $p$ and $q$ correspond to the active UL and DL RRH, respectively. Therefore, the SINR at the BBU is given by
\vspace{-0.2em}
\begin{align}\label{eq:SINRu SR}
 \SINRuS &=  \frac{P_u \ell(x_{p})|\qw_{r,p}^{\dag}\qg_{p}|^2  }
 { P_b \ell(x_p,x_{q})|\qw_{r,p}^{\dag}\qH_{\mathsf{ud}}^{pq}\qw_{t,q}|^2+ \|\qw_{r,p}\|^2}.
\end{align}

In the next section, we consider different processing schemes for transmit and receive beamforming vectors and characterize the system performance using the UL and DL average sum rate given by
\vspace{-0.2em}
\begin{align}
\RFD  = \mathcal{R}_\Us + \mathcal{R}_\Ds,\label{eq:achievable rate FD}
\end{align}
where \small{$\mathcal{R}_\Us =\E{\ln\left(1+\mathsf{SINR_u^i}\right)}$, $\mathcal{R}_\Ds =\E{ \ln \left(1+\mathsf{SINR_d^i}\right)}$ }\normalsize with $\mathsf{i}\in\{\mathsf{A}, \mathsf{S}\}$ are the spatial average UL and DL rates, respectively.
\section{Performance Analysis}
In this section, UL/DL average rates provisioned under the considered RRH association schemes are evaluated. We also present UL/DL average rates for a half-duplex user, which serves as a benchmark for performance comparison and to illustrate the gains due to full-duplex operation.
\vspace{-0.5em}
\subsection{Average Downlink Rate}
We consider MRT processing at the DL RRHs and set $\qw_{t,i}=\frac{\qh_i}{\|\qh_i\|}$. In the sequel, we will investigate the average DL rate for the ARA and SRA schemes.

\emph{ARA Scheme:} In this case, the received SINR at $U$ is given by~\eqref{eq:SINRd AR}.
For notational convenience, we denote $\delta=\frac{2}{\alpha}$, $X=P_b\sum_{x_i \in \Phi_\Ds \cap b(o,R)}X_i $ with $ X_i=\ell(x_{i})\|\qh_{i}\|^2$ and $Y={P_u}|h_{\mathsf{LI}}|^2$.  The following proposition provides the average DL rate achieved by the full-duplex user with the ARA scheme and MRT processing.

\begin{proposition}\label{Propos:Rd}
The average DL rate achieved by the ARA scheme with MRT processing can be expressed as\footnote{The average DL rate is zero for the case of $N_\Ds=0$.}
\vspace{-0.2em}
 \begin{align}\label{eq:Rd general}
 &\mathcal{R}_\Ds=\sum_{N_\Ds=1}^{\infty}
 \Bigg(
 \sum_{k=1}^{N_\Ds}\frac{(-1)^k}{k!}
 \underbrace{\sum_{n_1=1}^{N_\Ds}
 \cdots
 \sum_{n_k=1}^{N_\Ds}}_{n_1\neq n_2\cdots\neq n_k}\\
 &\times \int_{0}^{\infty}
 \frac{\exp(-z)}{z (1 + P_u\Sap z)}\prod_{\ell=1}^{k}
  {\mathcal{X}_{n_\ell}}(P_bz)dz\Bigg)\frac{\mu_\Ds^{N_\Ds}\exp(-\mu_\Ds)}{\Gamma(N_\Ds+1)},\nonumber
\end{align}
where
\vspace{-0.2em}
\begin{align}
\mathcal{X_\ell}(s)
&=  \frac{\delta}{R^2}\sum_{i=0}^{M-1}
\sum_{j=0}^i\!
\binom {i}{j}
\frac{ \epsilon^{i-j} }{\Gamma(i+1)}\frac{s}{(s+1)^{i-\left(j+\delta\right)+1}}\nonumber\\
&\hspace{1em}\times
G_{2 2}^{1 2} \left( \frac{R^{\alpha}}{s+1} \  \Big\vert \  {j+\delta-i, 1 \atop j+\delta,0} \right).\nonumber
\end{align}
\end{proposition}
\begin{proof}
See Appendix~\ref{Appendix:Proof RD}.
\end{proof}

We remark that the average DL rate is an increasing function of the cell radius. This follows from the fact that
increasing the cell radius also increases the effective density (and consequently the number) of the RRHs which serve the user. However, the gains become marginal after a certain value of $R$, since the received power from distant RRHs becomes negligible. In fact, it can be shown that as $R$ attains a large value, the average DL rate saturates and becomes independent of $R$  (cf. Section~\ref{sec:numerical}). Therefore, we let $R\rightarrow\infty$, which corresponds to case where all DL RRHs of $\Phi_{\Ds}$ participate in DL transmissions, since it allows us to conduct our analysis in an amicable way to present useful insights into the performance of the considered network. A similar assumption can be found in~\cite{Peng:WCL:2014}. The following proposition establishes an upper bound to the average DL rate, $\mathcal{\bar{R}}_d$ for non-singular and standard singular path loss models.
\begin{proposition}
The average DL rate achieved by the ARA scheme with MRT processing can be upper bounded as
\vspace{-0.2em}
 \begin{align}\label{eq:Rd general upper}
 \mathcal{\bar{R}}_\Ds &\!= \!\int_{0}^{\infty}\!\!\! \left(1\! -\! \exp\left(\!\!-2\pi p \lambda \int_{0}^{\infty}\!\!\! \left(1\!-\!\left(\!1 \!+\! \frac{zP_b}{\epsilon\!+\!\|x\|^{\alpha}}\!\right)^{\!\!\!-M}\right)dx\!\right)\!\right)\nonumber\\
 &\hspace{9em}\times\frac{\exp(-z)}{z (1 + P_u\Sap z)}dz\!.
\end{align}
Moreover, for $\epsilon \rightarrow 0$ (i.e., the standard singular path loss model) the average DL rate can be upper bounded as
\vspace{-0.3em}
\begin{align}\label{eq:Rd singular}
\mathcal{\bar{R}}_\Ds &\!=\!\sum_{k=1}^{\infty}\frac{( \Galf P_b^{\delta} )^{k}}{\Gamma(k+1)} G_{2 1}^{1 2} \left( P_u \Sap  \!\!\  \Big\vert \ \!\! {1 \!-\! \delta k,0 \atop 0} \right)\!,
\end{align}
where $\Galf  = \frac{\delta\pi p\lambda}{\Gamma(M)} \Gamma\left(M + \delta\right)\Gamma\left(- \delta\right)$.
\end{proposition}
\begin{proof}
By using~\cite[Lemma 1]{Hamdi:COM:2010}, $\mathcal{R}_\Ds$ can be expressed as
\vspace{-0.3em}
\begin{align}\label{eq:Rd proof init}
 \mathcal{R}_\Ds
 &=\int_{0}^{\infty} \mathcal{M}_Y(z)\left(1 - \mathcal{M}_X(z)\right)\frac{\exp(-z)}{z}dz,
\end{align}
where $\mathcal{M}_Y(s) = \frac{1}{1 +P_u\Sap s }$ and
\vspace{-0.3em}
\begin{align}\label{eq:Rd laplaces}
  &\mathcal{M}_X(s) ={\tt E}_{\Phd}\left\{ \exp\left(-s \sum_{i\in\Phd} \frac{P_b\|\qh_{i}\|^2}{\epsilon+\|x_i\|^{\alpha}}\right)\right\}\nonumber
\\
 &\quad \stackrel{(a)}{=}{\tt E}_{\Phd}\left\{ \prod_{i\in\Phd} {\tt E}_{\qh}\left\{\exp\left(- \frac{sP_b\|\qh\|^2}{\epsilon+\|x_i\|^{\alpha}}\right)\right\}\right\}
 \nonumber\\
 &\quad\stackrel{(b)}{=}\exp\!\left(\!-2\pi p \lambda \!\!\int_{0}^{\infty}\!\!\!
 \left(1\!-\!\left(\!1 \!+\! \frac{sP_b}{\epsilon+\|x\|^{\alpha}}\!\right)^{\!\!-M}\right)dx\!\right)\!.
\end{align}
In \eqref{eq:Rd laplaces} (a) follows from the fact that $\|\qh_{i}\|^2$ are i.i.d and also
independent from the point process $\Phd$ and (b) holds due to the probability generating functional
(PGFL) for a PPP~\cite{StochasticGeometry_Book_1996} and by using the MGF of $\|\qh\|^2$ which is chi-square distributed with $2M$ degrees of freedom.\footnote{In what follows, we will use the notation $x \sim\chi^2_{2K}$ to denote that $x$ is a chi-square distributed RV with $2K$ degrees-of-freedom.}

By converting the integral from Cartesian to polar coordinates, $\mathcal{M}_X(s)$ in~\eqref{eq:Rd laplaces} for $\epsilon\!\!\rightarrow\!0$ can be further simplified as
\vspace{-0.6em}
\begin{align}\label{eq:Rd laplaces col}
  \mathcal{M}_X(s) &\! =\exp\left(\Galf
(sP_b)^{\delta}\right).
\end{align}
Accordingly, by substituting~\eqref{eq:Rd laplaces col} into~\eqref{eq:Rd proof init} we obtain
\vspace{-0.3em}
\begin{align}\label{eq:Rd M comp}
 \mathcal{\bar{R}}_d &\!=\! \int_{0}^{\infty} \!\!\!\frac{\left(\!1 \!-\! \exp\left(\Galf(zP_b)^{\delta}\right)\!\right)\exp(-z)}{z (1\! +\! P_u\Sap z)}dz\!.
\end{align}
In order to simplify~\eqref{eq:Rd M comp}, we adopt a series expansion of the exponential term. Substituting the series expansion of $\exp\left(\Galf(zP_b)^{\delta}\right)$ into~\eqref{eq:Rd M comp} and then using $\frac{1}{1+cx^k} = G_{1 1}^{1 1} \left( cx^k \  \vert \  {0 \atop 0} \right)$, yields
\vspace{-0.2em}
\begin{align}\label{eq:Rd M}
 \mathcal{\bar{R}}_\Ds &= \sum_{k=1}^{\infty}
 \frac{(P_b^{\delta}\Galf )^{k}}{\Gamma(k+1)}\\
 &\quad\times\int_{0}^{\infty} z^{\delta k-1} \exp(-z) G_{1 1}^{1 1} \left( P_u \Sap z \  \Big\vert \  {0 \atop 0} \right)dz.\nonumber
\end{align}
To this end, using \cite[Eq. (7.813.1)]{Integral:Series:Ryzhik:1992} we
obtain the closed-form expression for $\mathcal{\bar{R}}_\Ds$ as given in \eqref{eq:Rd singular}.
\end{proof}

\emph{SRA Scheme:} For this scheme, the average DL rate, $\mathcal{R}_d$ is given in the following proposition.

\begin{proposition}
The average DL rate achieved by the SRA scheme with MRT processing can be expressed as
\vspace{-0.2em}
\small{
\begin{align}\label{eq:Rd SRA propos}
 &\mathcal{R}_\Ds=\sum_{N_\Ds=1}^{\infty}\frac{\mu_\Ds^{N_\Ds}\exp(-\mu_\Ds)}{\Gamma(N_\Ds+1)}
\int_{0}^{R} \left(e^{\!\frac{r^\alpha \!+\! \epsilon}{P_b}}
\sum_{n=0}^{M-1}\!\!
A_n
E_{M-n}\!\left(\!\frac{r^\alpha \!+\! \epsilon}{P_b}\!\right)
\right.\nonumber\\
&\hspace{7em}\left.+~\!
e^{\!\frac{1}{P_u\Sap}\!} B_0 E_1\!\left(\!\frac{1}{P_u\Sap}\right)\!\right) f_{\|x_{q}\|}(r)dr,
\end{align}}\normalsize
where $E_n(\cdot)$ is exponential integral~\cite[Eq. (8.211)]{Integral:Series:Ryzhik:1992}, \small{$A_n = \lim_{z\rightarrow-\frac{(r + \epsilon)^\alpha}{P_b}}\frac{1}{n!}\left(\frac{r ^\alpha\!+\! \epsilon}{P_b}\right)^n\frac{d^n}{dz^n} \frac{-1}{z(1+P_u\Sap z)},$ }\normalsize  and \small{$B_0=\Big(\!
\Big(1 \!- \!\frac{P_b}{P_u} \frac{1}{\Sap(r^\alpha\!+\!\epsilon)}\Big)^{-M}-\!1\Big).$ }\normalsize Moreover,  \vspace{-0.2em}
\small{
\begin{equation}\label{eq:pdf r}
 f_{\|x_{q}\|}(r) \!=\!\frac{2 N_\Ds}{r} \left(1 \!- \!\left(\frac{r}{R}\right)^2\right)^{N_\Ds-1}\!\!\left(\frac{r}{R}\right)^2,~ 0\!\leq r \!\leq R,
\end{equation}}\normalsize
\end{proposition}
\begin{proof}
The proof is omitted due to space limitations.
\end{proof}


\vspace{-0.7em}
\subsection{Average Uplink Rate}
In this subsection, we investigate the average UL rate with MRC/MRT and ZF/MRT processing respectively. In case of the ARA scheme, deriving the statistics of the UL SINR in~\eqref{eq:SINRu AR} with MRC/MRT and ZF/MRT appears intractable. Hence, in order to evaluate the average UL rate, we have resorted to simulations in Section~\ref{sec:numerical}. In the sequel, we consider the standard singular path loss model and obtain analytical expressions for the average UL rate.

\emph{SRA Scheme with MRC/MRT Processing:} MRC processing for the UL with MRT processing for the DL is the optimal
transmit-receive diversity technique since it can maximize the SNR. Although MRC/MRT processing is not optimal in presence of interference between the UL/DL RRHs, it could be favored in practice, because it can balance the performance and system complexity.

Substituting $\qw_{r,p}^{\MRC}=\frac{\qg_p}{\|\qg_p\|}$ and $\qw_{t,q}^{\MRT}$ into~\eqref{eq:SINRu SR}, the received SINR at the BBU can be expressed as
\vspace{-0.3em}
\begin{align}\label{eq:SINRu SR2}
 \mathsf{SINR_u} &=  \frac{P_u \ell(x_{p})\|\qg_{p}\|^2}
 { P_b\ell(x_p,x_{q})\sum_{i=1}^M Z_i+ 1},
\end{align}
where $Z_i  = U_i V_i$  with $U_i  = |\qw_{r,p}^{\MRC^\dag}\qh_{\mathsf{ud}i}^{pq}|^2$ and $V_i = (w_{t,q,i}^{\MRT})^2$ where $\qh_{\mathsf{ud}i}^{pq}$ is the $i$th column of $\vH_{\Us\Ds}^{pq}$ (i.e., $\vH_{\Us\Ds}^{pq} = [\vh_{\mathsf{ud}1}^{pq},\vh_{\mathsf{ud}2}^{pq},\cdots,\vh_{\mathsf{ud}M}^{pq}]$) and $w_{t,q,i}^{\MRT}$ is the $i$th element of $\qw_{t,q}^{\MRT}$.
For the notational convenience, let us denote $W = P_u \ell(x_{p})\|\qg_{p}\|^2 $, $Z =P_b\ell(x_p,x_{q})\sum_{i=1}^M Z_i $, and $d_{\Us\Ds}^{-\alpha} = \ell(x_p,x_q)$. As we assume the UL and DL RRHs are randomly positioned in the disk with radius $R$, the pdf $f_{d_{\Us\Ds}}(r)$ is given by~\cite{Moltchanov:Distance}
\begin{align}\label{eq:f_dr}
f_{d_{\Us\Ds}}(r) \!\! =\! \!\frac{2r}{R^2}\left(\!\frac{2}{\pi}\cos^{-1}\left(\!\frac{r}{2R}\!\!-\!\frac{r}{\pi R}\sqrt{1 \!-\!\! \frac{r^2}{4R^2}}\right)\!\!\right),
\end{align}
for $0<r<2R$. We now characterize the cdfs of $Z_i$ and $W$ in the following lemma which will be used to establish the average UL rate due to MRC/MRT processing.

\begin{Lemma}~\label{lemma:cdf Zi}
Let $\alpha = \frac{m}{n}$ with $\gcd(m,n)=1$ where $\gcd(m,n)$ is the greatest common divisor of integers $m$ and $n$. Then, the cdf of $W$ can be derived as
\vspace{-0.4em}
\begin{align}\label{eq:cdf W}
F_{W}(w)= 1 -\zeta
G_{m+1~2n+1}^{2n+1\quad m} \left(  \varsigma w^{2n} \  \Big\vert \ \!\!\!
 {\Delta(m,0),1 \atop  \Delta(2n,M),0 } \!\right)\!,
\end{align}
where $\zeta=\pi\frac{(2n)^{M}}{\Gamma(M)}\sqrt{\frac{2m}{(2\pi)^{m+2n}}}$, $\varsigma=\Big(\!\frac{1}{2nP_u}\!\Big)^{\!2n}\!\Big(\!\frac{m}{(1\!-\!p)\lambda\pi}\!\Big)^{\!\!m}$, and $\Delta(a,b) =\frac{b}{a},\cdots,\frac{a+b-1}{b}$.

Moreover, the cdf of $Z_i$ can be expressed as
\vspace{-0.4em}
\begin{align}\label{eq:cdf Zi}
F_{Z_i}(z) & =    G_{3 4}^{3 1} \left( \Sap z \  \Big\vert \  {1, M, M \atop 1, 1, M, 0} \right).
\end{align}
\end{Lemma}

\begin{proof}
The proof is omitted due to space limitations.
\end{proof}

\begin{proposition}\label{Propos:Rc:MRC}
The average UL rate achieved by the SRA scheme with MRC/MRT processing can be expressed as
\vspace{-0.4em}
\begin{align}\label{eq:Ru:MRC}
&\hspace{0em}\mathcal{R}_\Us\!= \! \mu\!
\int_{0}^{\infty}\!\!\!
\int_0^{2R}\!\!\!
G_{v u}^{u t}  \left(\left(\!\frac{2n\varsigma}{z}\!\right)^{\!\!2n}\!\!\! \
\Big\vert \ \!\!\! {\Delta(2n,0),\Delta(1,\Delta(m,0)),1 \atop \Delta(1,\Delta(2n,M)),0 } \!\right)\nonumber\\
&\hspace{0em}\times\!\!
\left(\!G_{4 4}^{3 2} \left( \frac{\Sap r^{\alpha}}{P_bz} \!\!\  \Big\vert \  \!\!\!{0,1, M, M \atop 1, 1, M, 0} \right)\!\right)^{\!\!M}\!\!
\frac{\exp(-z)}{z}f_{d_{\Us\Ds}}(r)dr dz\!,
\end{align}
where $\mu = 2\zeta \frac{\sqrt{n\pi}}{(2\pi)^n}$, $t=m+2n$, $v=t+1$,  $u=2n+1$  and $f_{d_{\Us\Ds}}(r)$ is given in~\eqref{eq:f_dr}.
\end{proposition}

\begin{proof}
See Appendix~\ref{Appendix:Proof RU MRC}.
\end{proof}

The MRC/MRT scheme does not take into account the impact of the interference between the UL/DL RRHs. Therefore, the system performance suffers under the impact of strong interference. Motivated by this, we now study the performance of a more sophisticated linear combining scheme, namely the ZF/MRT scheme.

\emph{SRA Scheme with ZF/MRT Processing:} We can adopt ZF beamforming at the UL RRH to completely cancel the interference between the UL/DL RRHs. To ensure this is possible, the number of the antennas equipped at the UL RRH should be greater than one, i.e., $M>1$. After substituting $\qw_{t,q}^{\MRT} = \frac{\qh_q}{\|\qh_q\|}$  into~\eqref{eq:SINRu AR}, the optimal receive beamforming vector at the UL RRH $\qw_{r,p}$ can be obtained by solving the following problem:
\vspace{-0.2em}
   \bea\label{eqn:wt}
    \max_{\|\qw_{r,p}\|=1} &&\hspace{1em}  |\qw_{r,p} \qg_{p}|^2 \nonumber\\
     \mbox{s.t.} &&\hspace{1em} \qw_{r,p}^{\dag}\qH_{\mathsf{ud}}^{pq}\qh_q =0.
 \eea
Hence, the optimal combining vector $\qw_{r,p}$  can be obtained as $\qw_{r,p}^{\ZF} = \frac{\qA \qg_{p}}{\|\qA \qg_{p}\|},$
where $\qA\triangleq \qI - \frac{\qH_{\mathsf{ud}}^{pq}\qh_q \qh_q^{\dag}\qH_{\mathsf{ud}}^{pq \dag} }{ \| \qH_{\mathsf{ud}}^{pq}\qh_q\|^2}$. Accordingly, substituting $\qw_{r,p}$  into~\eqref{eq:SINRu SR} the received SNR at the BBU can be expressed as
\vspace{-0.4em}
\begin{align}\label{eq:SINRu SR ZF}
 \SNRu &=  P_u \ell(x_{p})\|\tilde{\qg}_{p}\|^2,
\end{align}
where $\|\tilde{\qg}_{p}\|^2\sim \chi^2_{2(M-1)}$.

With the SNR, \eqref{eq:SINRu SR ZF} in hand, we now study the average UL rate of the SRA scheme with ZF processing and for any arbitrary value of $\alpha = \frac{m}{n}>2$  with $\gcd(m,n)=1$.

\begin{proposition}\label{Propos:Rc:ZF}
The average UL rate achieved by the SRA scheme with ZF/MRT processing can be expressed as
\vspace{-0.2em}
 \begin{align}\label{eq:Ru propos}
 \mathcal{R}_\Us =\kappa
 G_{v s}^{s t} \left( \varsigma  \  \Big\vert \  {\Delta(m,0),\Delta(2n,0),1 \atop \Delta(2n,M-1),\Delta(2n,0) }\right)\!\!,
\end{align}
where $\kappa=\frac{(2n)^{M-1}}{2\Gamma(M-1)}\sqrt{\frac{2m}{(2\pi)^{m+2n}}}$  and $s=4n+1$.

\end{proposition}

\begin{proof}
The proof is omitted due to space limitations.
\end{proof}
\vspace{-0.4em}
\subsection{Half-Duplex Transmission}
In this subsection, we compare the performance of half-duplex and full-duplex modes of operation at the user under the so called ``RF chain preserved" condition.\footnote{RF chains have a higher cost than antenna elements and therefore full-duplex/half-duplex studies based on RF chain preserved condition as compared to ``antenna-preserved'' condition has been widely accepted in the literature for fair comparison.} A half-duplex user employs orthogonal time slots for DL and UL transmissions, respectively. Consequently, with the ARA scheme and MRC/MRT precessing, the average sum rate of the half-duplex user is given by
\vspace{-0.0em}
\begin{align}\label{eq: sum rate of single-antenna HD AP}
\RHDs&\!=\tau {\tt E}\{\ln(1 \!+\! \SNRd )\} + (1\!-\!\tau){\tt E}\{\ln(1 \!+\!\SNRu  )\},
\end{align}
where $\tau$ is a fraction of the time slot duration of $T$, used for DL transmission, \small{$\SNRd=\sum_{i \in \Phi_\Ds \cap b(o,R)} P_b \ell(x_{i})|\qh_{i}^{\dag}\qw_{t,i}|^2  $ }\normalsize and \small{$\SNRu=\sum_{j\in \Phi_{\mathsf{u}}\cap b(o,R)} P_u \ell(x_{j})|\qw_{r,j}^{\dag}\qg_{j}|^2$. }\normalsize In this case, the average sum rate achieved by the ARA scheme can be obtained from~\eqref{eq:Rd singular}.
\begin{Corollary}
The average sum rate of the half-duplex user achieved by the ARA scheme can be approximated as
\vspace{-0.6em}
\begin{align}\label{eq:Rd M=1 Sap 0}
 \RHDs &\approx \tau \sum_{k=0}^{\infty}\!\frac{ (P_b^{\delta}\Galf )^{k}}{\Gamma(k+1)}
 \Gamma\left(\delta k\right)\nonumber\\
 &\quad + (1-\tau) \sum_{k=0}^{\infty}\!\frac{ (P_u^{\delta}\Galfu )^{k}}{\Gamma(k+1)}
 \Gamma\left(\delta k\right)\!.
 \end{align}
\end{Corollary}

\begin{proof}
The proof is omitted due to space limitations
\end{proof}

Note that since half-duplex transmissions does not suffer from LI and interference, the DL and UL SNR with SRA scheme can be found from~\eqref{eq:SINRu SR ZF}. Therefore, the average DL and UL rate achieved by the SRA scheme can be obtained by replacing $M$  and $1-p$  by $M+1$  and $p$  in~\eqref{eq:Ru propos}.
\begin{figure}[t]
\centering
\vspace{-1.2em}
\includegraphics[width=85mm, height=60mm]{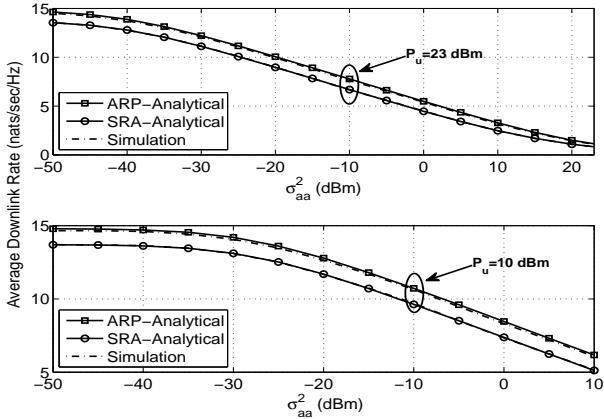}
\vspace{-1.1em}
\caption{Average DL rate of SRA and ARA schemes versus $\Sap$ ($P_b=46$ dBm, $M=2$, and $\lambda=0.001$).}
\vspace{-1.0em}
\label{fig: Average DL rate}
\end{figure}
\vspace{-0.3em}
\section{Numerical Results and Discussion}\label{sec:numerical}
In this section, we investigate the system performance and confirm the derived analytical results through comparison with Monte Carlo simulations. The simulations adopt parameters of a LTE-A network~\cite{LTE}. The maximum transmit power of the DL RRHs and the full-duplex user are set to $46$~dBm  and $23$ dBm, respectively. The receiver noise has a power spectral density of $-120$~dBm/Hz or $-50$ dBm  over the entire bandwidth of $10$ MHz.

Fig.~\ref{fig: Average DL rate} shows that average DL rate versus $\Sap$ for  $M=2$ and for the SRA and ARA schemes. We plot the average DL rate for two different power constraints $(P_b, P_u) = (46~\text{dBm}, 23~\text{dBm})$  and  $(P_b, P_u) = (46~\text{dBm}, 10~\text{dBm})$  and let the LI power vary between $-50$~\text{dBm}  and $P_u$~\text{dBm}.\footnote{With $P_u$~\text{dBm} we mean that no LI cancellation is applied at the full-duplex user. Employing different passive and digital cancellation methods, some practical full-duplex radios can essentially cancel the LI almost to the noise floor~\cite{Sabharwal:JSAC:2014, Katti:Sigcomm:2013}.}
The analytical upper bounds for the average DL rate of SRA and ARA scheme are also included which are sufficiently tight.  As we observe when the $P_u$ is low, the ARA scheme consistently outperforms the SRA scheme in all regimes of LI strength. However, it is clear that the gap between the ARA and SRA scheme decrease when the LI strength increase (i.e., both $\Sap$ and $P_u$ are high) and becomes negligible when no LI cancellation is applied. On the other hand, although increasing the transmit power of the full-duplex user $P_u$ increases the average UL rate of the system, (cf. Fig.~\ref{fig: Average UL rate}) it degrades the average DL rate.

Fig.~\ref{fig: Average UL rate} compares the average UL rate of the SRA scheme with MRC/MRT and ZF/MRT processing and under different cases of user power and path loss exponent values. It can be observed that the analytical curves are in perfect agreement with the simulations. In addition, the average UL rate due to the MRC/MRT processing degrades when the interference power from the DL RRH becomes stronger (i.e., when $P_b$ increases), while the average UL rate due to ZF/MRT processing remains the same regardless of the interference power level. Moreover, we see that the MRC/MRT  outperforms ZF/MRT in the low interference power regime.

\begin{figure}[t]
\centering
\vspace{-1.4em}
\includegraphics[width=85mm, height=60mm]{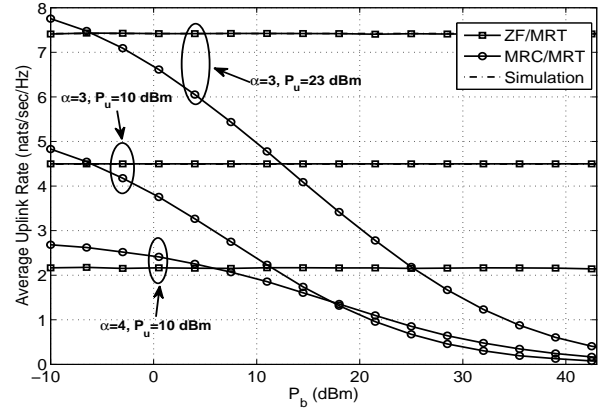}
\vspace{-0.9em}
\caption{Average UL rate of the SRA scheme with MRC/MRT and ZF/MRT processing ($M=2$, $p=0.5$, and $\lambda=0.001$).}
\vspace{-1.0em}
\label{fig: Average UL rate}
\end{figure}
\begin{figure}[t]
\centering
\vspace{-0.5em}
\includegraphics[width=85mm, height=60mm]{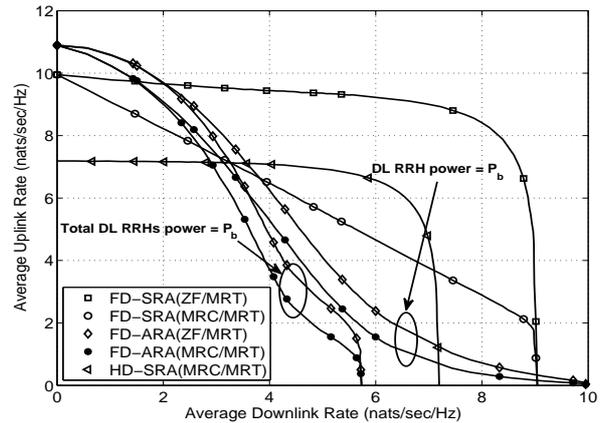}
\vspace{-0.4em}
\caption{Rate region of the ARA and SRA schemes for full-duplex and half-duplex modes of operation ($M=3$, $\alpha=3$, and $\lambda=0.001$).}
\vspace{-1.2em}
\label{fig: rate region}
\end{figure}
Fig.~\ref{fig: rate region} shows the rate region of the ARA and SRA schemes respectively for both full-duplex and half-duplex modes of network operation. In this figure, we have set $P_u=23$ dBm, $\Sap=-30$ dBm and change $p$ from $0$ (i.e. only UL transmission) to $1$ (i.e. only DL transmission). For a fair comparison between the ARA and SRA schemes, we have also included the case where the same total transmit power constraint is imposed on the DL such that the transmit power of the single DL RRH in SRA scheme ($P_b=23$ dBm) is equally divided among all the DL RRHs in the ARA scheme. For the ARA scheme with ZF/MRT processing we assume that each UL RRH adjusts its receive beamforming vector in such a way that the interference from its nearest DL RRH is canceled. These results reveal that the ARA scheme results in a rate region that is strongly biased towards UL or DL, but using the SRA scheme results in a more balanced rate region. For this setup, SRA scheme with ZF/MRT processing can achieve up to $30\%$ and $39\%$ average sum rate gains as compared to the half-duplex SRA and full-duplex ARA scheme counterparts, respectively.
\vspace{0em}
\section{Conclusion}
In this paper, we studied the average sum rate of a C-RAN with randomly distributed multiple antenna UL and DL RRHs communicating with a full-duplex user. Specifically, the performance of two RRH association schemes, namely, ARA and SRA with MRC/MRT and ZF/MRC processing were studied and analytical expressions for the average UL and DL rates were derived. The SRA scheme achieves a superior performance as compared to the ARA scheme. We found that for a fixed value of LI power, the SRA scheme with ZF/MRT processing can ensure a balance between maximizing the average sum rate and maintaining an acceptable fairness level between UL/DL transmissions. Our results show that full-duplex transmissions can achieve higher data rates as compared to half-duplex mode of operation, if proper RRH association and beamforming are utilized and the residual LI is sufficiently small.
\appendices
\vspace{-0.5em}
\section{Proof of Proposition~\ref{Propos:Rd}}
\label{Appendix:Proof RD}
With the aid of~\cite[Lemma 1]{Hamdi:COM:2010}, the average DL rate conditioned on the number of DL RRHs inside the cell can be written as
\vspace{-0.4em}
\begin{align}\label{eq:Rd def}
 \mathcal{R}_\Ds &= {\tt E}\left\{\ln\left(1 + \frac{X}{Y+1}\right)\Big\vert N_\Ds\right\}\\
 &=\sum_{N_\Ds=1}^{\infty}\left(\int_{0}^{\infty}\!\!\! \mathcal{M}_Y(z)\left(1 \!-\! \mathcal{M}_X(z)\right)\frac{\exp(-z)}{z}dz\right)\Prob(N_\Ds),\nonumber
\end{align}
where $\mathcal{M}_Y(s) = \frac{1}{1 +P_u\Sap s }$.  In~\eqref{eq:Rd def} second equality holds since $N_\Ds$ is a Poisson RV. Moreover, since the channels are assumed to be i.i.d, the MGF of $X$ can be expressed as $\mathcal{M}_{X}(s) = \prod_{\ell=1}^{N_\Ds} \mathcal{M}_{X_\ell}(P_bs)$.  Using the differentiation property of the Laplace transform, $\mathcal{M}_{X_\ell}(s)$  can be written as
$\mathcal{M}_{X_\ell}(s)=s \mathcal{L}\left(F_{X_\ell} (x)\right),$

where $\mathcal{L}(\cdot)$ denotes the Laplace transform and $F_{X_\ell} (x)$ is given by~\cite{Ratnarajah:TSP:2015}
\vspace{-0.4em}
\begin{align}
F_{X_\ell} (x)  \!=\! 1 \!-\! \frac{\delta}{ R^2}
\sum_{i=0}^{M-1}\!\!
\sum_{j=0}^i
\binom {i}{j}
\frac{\epsilon^{i-j} x^{i-\left(j+\delta\right)}}{\Gamma(i+1)}e^{-x}\gamma\left(j\!+\!\delta, x R^{\alpha}\right)\!.\nonumber
\end{align}
Now by using the identity $\gamma(\nu,x) = G_{1 2}^{1 1} \left( x \  \vert \  {1 \atop \nu, 0} \right),$

we get
\vspace{-0.4em}
\begin{align}\label{eq:MXi:semifinal}
\mathcal{M}_{X_\ell}(s)  \!&= 1\! -\! \frac{\delta}{ R^2}\sum_{i=0}^{M-1}
\sum_{j=0}^i\!
\binom {i}{j}
\frac{\epsilon^{i-j} }{\Gamma(i+1)}\\
&\hspace{-1em}\times s\int_{0}^{\infty}e^{-(s+1)x} x^{i-\left(j+\delta\right)}
G_{1 2}^{1 1} \left(  x R^{\alpha} \  \Big\vert \  {1 \atop j\!+\!\delta, 0} \right)dx,\nonumber
\end{align}
which can be evaluated with the help of~\cite[Eq. (7.813.1)]{Integral:Series:Ryzhik:1992} to
yield
\vspace{-0.6em}
\begin{align}\label{eq:MXi:final}
\mathcal{M}_{X_\ell}(s)
&= 1 \!-\! \frac{\delta}{R^2}\sum_{i=0}^{M-1}
\sum_{j=0}^i\!
\binom {i}{j}
\frac{ \epsilon^{i-j} }{\Gamma(i+1)}\frac{s}{(s+1)^{i-\left(j+\delta\right)+1}}\nonumber\\
&\hspace{1em}\times
G_{2 2}^{1 2} \left( \frac{R^{\alpha}}{s+1} \  \Big\vert \  {j+\delta-i, 1 \atop j+\delta ,0} \right).
\end{align}
To this end, substituting~\eqref{eq:MXi:final} into~\eqref{eq:Rd def}, after some algebraic manipulations we obtain the desired result in~\eqref{eq:Rd general}.
\vspace{-0.4em}
\section{Proof of Proposition~\ref{Propos:Rc:MRC}}
\label{Appendix:Proof RU MRC}
Conditioned on $\ell(x_p,x_{q})$,  the RVs $W$  and $Z$  are independent. Hence we have
\vspace{-0.6em}
\begin{align}\label{eq:Rd proof}
 \mathcal{R}_\Us
 &\!=\!\int_0^{2R}\!\!
 \int_{0}^{\infty}\!\!
 \frac{ \mathcal{M}_Z(z)\left(1 \!\!- \!\mathcal{M}_W(z)\right)e^{-z}}{z} f_{d_{\Us\Ds}}(r)dr dz.
\end{align}

Therefore, we need to compute the Laplace transforms $\mathcal{M}_Z(s)$ and $\mathcal{M}_W(s)$  to derive the average UL rate. Note that $\mathcal{M}_Z(s) = \prod_{i=1}^M \mathcal{M}_{Z_i}(P_b d_{\Us\Ds}^{-\alpha} s)$.  Using the differentiation property
of Laplace transform i.e., $\mathcal{M}_{Z_i}(s) = s \mathcal{L}(F_{Z_i}(x))$ and $F_{Z_i}(x)$ from Lemma~\ref{lemma:cdf Zi}  and then applying the integral equality~\cite[Eq. (3.40.1)]{Laplace:Prudnikov:1992} we obtain
\vspace{-0.6em}
\begin{align}\label{eq:laplac:Zi}
\mathcal{M}_{Z_i}(s) =G_{4 4}^{3 2} \left( \frac{\Sap }{s} \  \Big\vert \  {0,1, M, M \atop 1, 1, M, 0} \right).
\end{align}

Using the differentiation property of Laplace transform and Lemma~\ref{lemma:cdf Zi}, \small{$\mathcal{M}_W(s)$ }\normalsize can be obtained as
\vspace{-0.6em}
\begin{align}\label{eq:laplac:W}
\mathcal{M}_W(s) \!=  \!1 \!-\!
\mu G_{u4}^{4t}  \left(\!\!\left(\!\frac{2n\varsigma}{s}\!\right)^{\!\!2n}\!\!\!\! \
\Big\vert \ \!\!\! {\Delta(2n,0),\Delta(1,\Delta(m,0)),1 \atop \Delta(1,\Delta(2n,M)),0 }\! \right)\!.
\end{align}
To this end, substituting~\eqref{eq:laplac:Zi} and~\eqref{eq:laplac:W} into~\eqref{eq:Rd proof} yields the desired result in~\eqref{eq:Ru:MRC}, thus completing the proof.
\bibliographystyle{IEEEtran}


\end{document}